\documentclass[11pt]{article}%
\usepackage{amsmath,amsfonts,amssymb,amsthm}
\usepackage{graphicx}
\usepackage[pagebackref]{hyperref}
\usepackage{pgfplots}
\usepackage{enumerate}
\usepackage{amsmath}
\usepackage{amsfonts}
\usepackage{amssymb}%

\usepackage[hmargin=1in,bmargin=1.35in,tmargin=1in]{geometry}

\providecommand{\U}[1]{\protect\rule{.1in}{.1in}}
\DeclareMathAlphabet{\mathbbold}{U}{bbold}{m}{n}
\hypersetup{
colorlinks=true,     linkcolor=blue,     citecolor=blue,     urlcolor=blue }
\pgfplotsset{compat=1.13}
\renewcommand{\backref}[1]{}

\renewcommand{\backrefalt}[4]{\ifcase #1 \or
[p.\ #2]\else
[pp.\ #2]\fi}
\makeatletter
\newcommand{\para}{ \@startsection{paragraph}{4} {\z@}{2ex \@plus 3.3ex \@minus .2ex}{-1em} {\normalfont\normalsize\bfseries}}
\makeatother
\providecommand{\U}[1]{\protect\rule{.1in}{.1in}}
\newtheorem{theorem}{Theorem}

\newtheorem{lemma}[theorem]{Lemma}

\let\oldproofname=\proofname
\renewcommand{\proofname}{\rm\bf{\oldproofname}}

\newcommand{\eps}{\varepsilon}

\newcommand{\mathify}[1]{\ifmmode{#1}\else\mbox{$#1$}\fi}

\begin{document}

\title{Quantum Approximate Counting, Simplified}
\author{Scott Aaronson\thanks{University of Texas at Austin. \ Email:
\texttt{aaronson@cs.utexas.edu}. \ Supported by a Vannevar Bush Fellowship
from the US Department of Defense, a Simons Investigator Award, and the Simons
\textquotedblleft It from Qubit\textquotedblright\ collaboration.}
\and Patrick Rall\thanks{University of Texas at Austin. \ Email:
\texttt{patrickjrall@gmail.com}. \ Supported by Aaronson's Vannevar Bush
Faculty Fellowship from the US Department of Defense.}}
\date{}
\maketitle

\begin{abstract}
In 1998, Brassard, H\o yer, Mosca, and Tapp (BHMT) gave a quantum algorithm
for \textit{approximate counting}. \ Given a list of $N$ items, $K$ of them
marked, their algorithm estimates $K$ to within relative error $\varepsilon
    $\ by making only $O\left( \eps^{-1} \sqrt{N/K} \right)$\ queries. \ Although this speedup is of \textquotedblleft
Grover\textquotedblright\ type, the BHMT algorithm has the curious feature of
relying on the Quantum Fourier Transform (QFT), more commonly associated with
Shor's algorithm. \ Is this necessary? \ This paper presents a simplified algorithm which we prove achieves the same query complexity using Grover
iterations only. \ We also generalize this to a QFT-free algorithm for
amplitude estimation. \ Related approaches to approximate counting were sketched previously by Grover, Abrams and Williams, Suzuki et al., and Wie (the latter two as we were writing this paper), but in all cases without rigorous analysis.
\end{abstract}

\begingroup
\renewcommand\thefootnote{}\footnote{Update November 2021: changed several constants throughout and gave an updated proof that simplifies the analysis. This also remedies an algebra mistake present in the previous version.}\addtocounter{footnote}{-1}\endgroup

\section{Introduction\label{INTRO}}

{\em Approximate counting} is one of the most fundamental problems in computer science. \ Given a list of $N$ items, of which $K>0$ are marked, the problem is to estimate $K$ to within a multiplicative error of $\varepsilon$. \ One wants to do this with the minimum number of {\em queries}, where a query simply returns whether a given item $i\in [N]$ is marked.

Two decades ago, Brassard, H\o yer, Mosca, and Tapp \cite{bhmt} gave a celebrated quantum algorithm for approximate counting, which uses $O\left( \eps^{-1} \sqrt{N/K} \right)$\ queries. \ This is tight, matching a lower bound of Nayak and Wu \cite{nayakwu}, and is a quadratic speedup over the best possible classical query complexity of $\Theta\left( \eps^{-2} (N/K)  \right)  $. \ This is the same type of quantum speedup as provided by the famous {\em Grover's algorithm} \cite{grover}, for {\em finding} a marked item in a list of size $N$, and indeed the BHMT algorithm builds on Grover's algorithm.

Curiously, though, the BHMT algorithm was {\em not} just a simple extension of Grover's algorithm to a slightly more general problem (approximate counting rather than search). \ Instead, BHMT made essential use of the Quantum Fourier Transform (QFT): a component that appears nowhere in Grover's algorithm, and that's more commonly associated with the exponential speedup of Shor's factoring algorithm \cite{shor}. \ Indeed, BHMT presented their approximate counting algorithm as a sort of hybrid of Grover's and Shor's algorithms.

This raises an obvious question: is the QFT in any sense {\em necessary} for the quadratic quantum speedup for approximate counting? \ Or can we obtain that ``Grover-like speedup'' by purely Grover-like means?

In this paper we settle that question, by giving the first rigorous quantum approximate counting algorithm that's based {\em entirely} on Grover iterations, with no QFTs or other quantum-mechanical ingredients. \ Matching \cite{bhmt}, the query complexity of our algorithm is the optimal $O\left( \eps^{-1} \sqrt{N/K} \right) $, while the computational complexity exceeds the query complexity by only an $O(\log N)$ multiplicative factor. \ Because of its extreme simplicity, our algorithm might be more amenable than \cite{bhmt} or other alternatives to implementing on near-term quantum computers. \ The analysis of our algorithm is also simple, employing standard classical techniques akin to estimating the bias of a coin via many coin tosses.\footnote{ Indeed, \cite{bhmt} also requires such classical techniques, since the QFT alone fails to reliably extract the desired information from the Grover operator. \ By removing the QFT, we show that Grover and estimation techniques alone can do the whole job. \ This gives a clear way to understand in what sense our algorithm is ``simpler.'' }

An approach broadly similar to ours was outlined by Grover \cite{grover-mean} in 1997, with fuller discussion by Abrams and Williams \cite{abrams-williams} in 1999. \ The latter authors sketched how to estimate the integral of a function over some domain, to additive error $\eps$, using $O(\eps^{-1})$ quantum queries. \ Crucially, however, neither Grover nor Abrams and Williams prove the correctness of their approach---among other issues, they assume that a probability can be estimated to a desired precision without any chance of failure. \ Also, it is not clear how to adapt their approaches to the broader problem of amplitude estimation.

As we were writing this paper, two other quantum algorithms for approximate counting were announced that avoid the use of QFTs. \ Surprisingly, both algorithms differ significantly from ours.

In April, Suzuki et al.\ \cite{suzuki}\ gave an $O\left( \eps^{-1} \sqrt{N/K} \right)$-query quantum algorithm that first collects samples from various Grover iterations, and then extracts an approximate value of $K$ via maximum likelihood estimation. \ Finding the maximum of the likelihood function, according to Suzuki et al., incurs a $\log \frac{1}{\eps}$ computational overhead. \ More importantly, even if we focus only on query complexity, Suzuki et al.\ do not prove their algorithm correct. \ Their analysis gives only a \emph{lower} bound on the error, rather than an upper bound, so is supplemented by numerical experiments. \ By contrast, our analysis is fully rigorous. \ On the other hand, the Suzuki et al.\ algorithm has the interesting feature that its invocations of Grover's algorithm are nonadaptive (i.e., can be performed simultaneously), whereas our algorithm requires adaptivity.

In July, Wie \cite{wie} sketched another $O\left( \eps^{-1} \sqrt{N/K} \right)$-query QFT-free quantum approximate counting algorithm. \ Wie's algorithm is based on Hadamard tests, which require the more expensive ``controlled-Grover'' operation rather than just bare Grover iterations. \ Replacing the QFT with Hadamard tests is called ``iterative phase estimation,'' and was suggested by Kitaev \cite{kitaev:meas}. \ Wie modifies iterative phase estimation in order to apply it to the BHMT algorithm. \ Unfortunately, and like the previously mentioned authors, Wie gives no proof of correctness. \ Indeed, given a subroutine that accepts with probability $p$, Wie (much like Abrams and Williams \cite{abrams-williams}) simply assumes that $p$ can be extracted to the requisite precision. \ There is no analysis of the overhead incurred in dealing with inevitable errors. \ Again, in place of analysis there are numerical experiments.

One reason why approximate counting is of interest in quantum computation is that it generalizes to a technique called {\em amplitude estimation}. \ Amplitude estimation is a pervasive subroutine in quantum algorithms, yielding for example faster quantum algorithms for mean estimation, estimation of the trace of high-dimensional matrices, and estimation of the partition function in physics problems \cite{montanaro:partition}. \ In general, amplitude estimation can upgrade almost any classical Monte-Carlo-type estimation algorithm to a quantum algorithm with a quadratic improvement in the accuracy-runtime tradeoff. \ Once we have our quantum approximate counting algorithm, it will be nearly trivial to do amplitude estimation as well.

In Section \ref{MAIN} we present our main result---the QFT-free approximate counting algorithm and its analysis---and then in Section \ref{AMPEST} we generalize it to amplitude estimation. \ We conclude in Section \ref{OPEN} with some open problems.

\subsection{Main Ideas}

Our algorithm for approximate counting mirrors a standard classical approach for the following problem: given a biased coin that is heads with probability $p$, estimate $p$. \ First, for $k = 0,1,2,\ldots$ the coin is tossed $2^k$ times, and stop at the $k$ where heads was observed at least once. \ This gives a rough guess for $p$, up to some multiplicative constant. \ Second, this rough guess is improved to the desired $1+\eps$ approximation via more coin tosses. 

Of course, we'd like to use Grover's algorithm \cite{grover} to speed up this classical approach quadratically. \ Grover's algorithm can be seen as a special `quantum coin,' which works as follows. \ If $K$ out of $N$ items are marked, then define $\theta := \arcsin \sqrt{K/N}$ to be the `Grover angle.' \ For any odd integer $r$, Grover's algorithm lets us prepare a coin that lands heads with probability $p = \sin^2(r\theta)$, by making $O(r)$ queries to the oracle.

The key idea of our algorithm is to use this `Grover coin' repeatedly, in a manner akin to binary search---adaptively varying the value of $r$ in order to zero in on the correct value of $\theta$ and hence $K$. \ In more detail, suppose that $\theta$ has already been narrowed down to the range $[\theta_\text{min},\theta_\text{max}]$. \ Then in a given iteration of the algorithm, the goal is to shrink this range by a constant factor, either by increasing $\theta_\text{min}$ or by decreasing $\theta_\text{max}$. \ To do so, we need to rule out one of the two possibilities $\theta \approx \theta_\text{min}$ or $\theta \approx \theta_\text{max}$. \ This, in turn, is done by finding some value of $r$ that distinguishes the two possibilities, by making $\theta \approx \theta_\text{min}$ and $\theta \approx \theta_\text{max}$ lead to two nearly-orthogonal quantum states that are easy to distinguish by a measurement.

But why should such a value of $r$ even exist---and if it does, why should it be small enough to yield the desired query complexity? \ Here we need a technical claim, which we call the ``Rotation Lemma'' (Lemma~\ref{lemma:rotation}). \ Consider two runners, who race around and around a circular track at differing constant speeds (corresponding to $\theta_\text{min}$ and $\theta_\text{max}$). \ Then informally, the Rotation Lemma upper-bounds how long we need to wait until we find one runner reasonably close to the start or the midpoint of the track, while the other runner is reasonably close to the one-quarters or three-quarters points. \ Here we assume that the ratio of the runners' speeds is reasonably close to $1$. \ We ensure this property with an initial preprocessing step, to find bounds $\theta_\text{min} \leq \theta \leq \theta_\text{max}$ such that $\theta_\text{max}/\theta_\text{min} \leq 1.65$.

Armed with the Rotation Lemma, we can zero in exponentially on the correct value of $\theta$, gaining $\Omega(1)$ bits of precision per iteration. \ The central remaining difficulty is to deal with the fact that our `Grover coin' is, after all, a {\em coin}---which means that each iteration of our algorithm, no matter how often it flips that coin, will have some nonzero probability of guessing wrong and causing a fatal error. \ Of course, we can reduce the error probability by using amplification and Chernoff bounds. \ However, amplifying na\"{i}vely produces additional factors of $\log(\frac{1}{\varepsilon})$ or $\log\log(\frac{1}{\varepsilon})$ in the query complexity. \ To eliminate those factors and obtain a tight result, we solve an optimization problem to find a carefully-tuned amplification schedule, which then leads to a geometric series for the overall query complexity.

\section{Approximate Counting\label{MAIN}}

We are now ready to state and analyze our main algorithm.

\begin{theorem} \label{thm:approxcount}  Let $S \subseteq [N]$ be a nonempty set of marked items, and let $K = |S|$. \ Given access to a membership oracle to $S$ and $\eps,\delta > 0$, there exists a quantum algorithm that outputs an estimate $\hat K$. The output $\hat K$ satisfies
    $$K(1-\eps) < \hat K < K (1+\eps)$$
    with probability at least $1-\delta$. There exists a function $Q(N,K,\eps,\delta) \in O\left(  \sqrt{N/K} \eps^{-1} \log(\delta^{-1}) \right)$ such that the algorithm makes fewer than $Q(N,K,\eps,\delta)$ queries whenever the estimate is accurate. The algorithm needs $O(\log N)$ qubits of space.
\end{theorem}
\begin{proof} The algorithm is as follows.\\[3mm]
\noindent\fbox{\parbox{\textwidth}{\hspace{3mm}\parbox{0.95\textwidth}{\hspace{1mm}\vspace{0mm}\\ 
    \textbf{Algorithm: Approximate Counting}\\
    \textbf{Inputs:} $\eps,\delta > 0$ and an oracle for membership in a nonempty set $S \subseteq [N]$.  \\
    \textbf{Output:} An estimate of $K = |S|$.

    \noindent  We can assume without loss of generality that $K \ll N$, for example by padding out the list with $10^6  N$ unmarked items. Let $\theta := \arcsin \sqrt{K/N}$; then since $K/N \leq  (10^6+1)^{-1} $, we have $\theta \leq 0.001$.

    Let $U$ be the membership oracle, which satisfies $U|x\rangle = (-1)^{x \in S} |x\rangle$. \ Also, let $|\psi\rangle$ be the uniform superposition over all $N$ items, and let $G := (I - 2|\psi \rangle\langle \psi |)U$ be the Grover diffusion operator. 
    \begin{enumerate}
        \item For $k := 0,1,2,\ldots$:
        \begin{enumerate}
            \item Let $r_k$ be the largest odd integer less than or equal to $1.05^k$. \ Prepare the state $G^{(r_k-1)/2} |\psi\rangle$ and measure. \ Do this at least $5000 \cdot\ln(5/\delta)$ times.
            \item If a marked item was measured $\geq 95\%$ of the time, exit the loop on $k =: k_\text{end}$.
        \end{enumerate}
    \item Initialize $\theta_\text{min} := 0.9 \cdot 1.05^{-k_\text{end}}  $ and $\theta_\text{max} := 1.65 \cdot  \theta_\text{min}  $. Then, for $t := 0,1,2,\ldots$:
        \begin{enumerate}
        \item  Use Lemma~\ref{lemma:rotation} to choose $r_t$.
        \item Prepare the state $G^{(r_t-1)/2} |\psi\rangle$ and measure. \ Do this at least $250 \cdot \ln\left( \delta_t^{-1} \right)$ times, where $\delta_t := (\delta \eps / 65) \cdot (0.9)^{-t}$.
        \item Let $\gamma := \theta_\text{max}/\theta_\text{min} - 1$.\ If a marked item was measured at least $12\%$ of the time, set $\theta_\text{min} := \theta_\text{max}/(1+0.9\gamma)$.\ Otherwise, set $\theta_\text{max} := (1+0.9\gamma)\theta_\text{min}$.
        \item If $\theta_\text{max} \le (1+\eps/5)\theta_\text{min}$ then exit the loop.
        \end{enumerate}
    \item Return $\hat K := N \cdot \sin^2\left(\theta_\text{max}\right)$ as an estimate for $K$.
    \end{enumerate}}}}\\ 

    The algorithm naturally divides into two pieces. \ First, step 1 computes an ``initial rough guess'' for the angle $\theta$ (and hence, indirectly, the number of marked items $K$), accurate up to {\em some} multiplicative constant, but not necessarily $1+\eps/5$. \ More precisely, step 1 outputs bounds $\theta_\text{min}$ and $\theta_\text{max}$, which are supposed to satisfy $\theta_\text{min} \le \theta \le \theta_\text{max}$ and satisfy $\theta_\text{max}/\theta_\text{min} \leq 1.65$. \ Next, step 2 improves this constant-factor estimate to a $(1+\eps/5)$-factor estimate of $\theta$, yielding a $(1+\eps)$-factor estimate of $K$.

Both of these steps repeatedly prepare and measure the following quantum state \cite{grover}:
    \begin{equation}G^{(r-1)/2}|\psi\rangle =  \frac{\sin(r\theta)}{\sqrt{K}} \sum_{x\in S} |x\rangle   +  \frac{\cos(r\theta) }{\sqrt{N-K}}  \sum_{x\not\in S} |x\rangle.  \end{equation}
        Note that, if this state is measured in the computational basis, then the probability of observing a marked item is $\sin^2(r\theta)$. This circuit requires $O(r)$ queries and needs $O(\log N)$ qubits to store $|\psi\rangle$.

    In what follows, we'll first prove that step 1 indeed returns a constant-factor $1.65$ approximation to $\theta$ with probability $\geq 1-\delta/2$. When it succeeds at this, we show that its query complexity is also $O\left( \sqrt{ N/K } \log\left( \delta^{-1} \right) \right)$. Next, we show that step 2 improves the estimate to a $(1+\eps/5)$-factor approximation with probability $\geq 1- \delta/2$, and deterministically requires an additional $O\left( \sqrt{ N/K } \eps^{-1} \log\left(\delta^{-1}\right) \right)$. The total failure probability is $\leq \delta/2 + \delta/2 = \delta$.

    \noindent \textit{Correctness of step 1.} First, we describe the ideal behavior: we want to terminate at a $k_\text{end}$ such that the resulting bounds $\theta_\text{min},\theta_\text{max}$ are accurate. Let $k_0$ be given by:
    \begin{align}
        k_0 := \text{the largest integer such that } \theta \cdot  1.05^{k_0} \leq 0.9
    \end{align}

Ideally, for values of $k$ satisfying $0 \leq k \leq k_0$, we do not terminate. Then we might terminate at $k_0+1$, $k_0+2$, etc., but if we reach $k_0+ 10$ then we definitely terminate there. Given that we have $k_0 + 1 \leq k_\text{end} \leq k_0 + 10$, we derive:
    \begin{align}
        k_\text{end} - 10  \leq  k_0 &\leq k_\text{end} - 1\\
       \theta \cdot 1.05^{k_\text{end}-10}  \leq \theta \cdot 1.05^{k_0}  &\leq 0.9 \leq \theta \cdot 1.05^{k_0 + 1} \leq \theta \cdot 1.05^{k_\text{end}} \\
        0.9 \cdot 1.05^{-k_\text{end}}   \leq \theta &\leq 0.9 \cdot  1.05^{-k_\text{end}} \cdot 1.05^{10}\\
        &\leq 0.9 \cdot 1.05^{-k_\text{end}} \cdot 1.65 
    \end{align}
 We have proved that in the ideal case we have $\theta \in [\theta_\text{min},\theta_\text{max}]$. Also see how $\theta_\text{max}/\theta_\text{min} = 1.65$. Next we move on to demonstrating that this ideal case occurs with probability $\geq 1-\delta/2$. The proof splits into two cases: terminating too late and terminating too early.

    We show that we do not terminate too late: if we reach $k = k_0  + 10$, then we terminate there with probability $\geq 1-  \delta/5$. Since $r_k$ is obtained by rounding downward to the nearest odd number, we can bound the rounding error:
    \begin{align}
        r_{k_0 + 10} &\geq 1.05^{k_0+10} - 2\\
        1.58 &\geq 0.9 \cdot 1.05^{10}  \geq 1.05^{k_0+10}\theta \geq 0.9\cdot 1.05^{10} \geq 1.396 \\
        \sin^2( r_{k_0+10} \theta ) &\geq \sin^2( (1.05^{k_0  + 10} - 2) \theta ) \geq  0.99 \cdot \sin^2(1.05^{k_0  + 10} \theta).
    \end{align}
    In the final equation we leverage that, on the interval $[1.396, 1.58]$, we have that $\sin^2(x)$ is increasing and that $\sin^2(x - 0.002) > 0.99\cdot\sin^2(x)$.  We also use that $\theta < 0.001$, and that the error in $\sin^2(r\theta)$ due to rounding $r$ only gets better with decreasing $\theta$.
    
    From this we conclude $ \sin^2(r_{k_0 + 10}\theta) \geq 0.99 \cdot \sin^2( 1.396 ) \geq 0.96 $. We see that:
    \begin{align}
        \text{Pr}[ \text{fail to terminate at } k = k_0 + 10  ] &=  \text{Pr}[ \text{fraction of marked items} \leq  95\% ] \\
        &\leq \exp\left( - 2 \cdot 5000 \cdot \left( 0.96 - 0.95 \right)^2 \ln(5/\delta) \right) \\
        &= \exp( - \ln(5/\delta)) = \delta/5
    \end{align}

    So all that remains to show is that we do not terminate too early: for all $k \leq k_0$ we do not see enough heads. Abbreviate $p_k = \sin^2(r_k \theta)$. For this calculation we demand the tighter version of the Chernoff-Hoeffding theorem: assuming $p_k < 95\%$, if we toss coin with bias $p_k$ a total of $m$ times, then the probability we see more than $95\%$ heads is bounded by:
    \begin{align}
 \text{Pr}[\text{fraction observed heads} \geq 95\%]  \leq  \exp( - m D( 95\% || p_k )  ) 
    \end{align}
where $D( 95\% || p_k)$ is the Kullback-Leibler divergence of two Bernoulli trials, given by:
\begin{align}
    D( 95\% || p_k) = (95\%) \ln \frac{95\%}{p_k} + (1- 95\%) \ln \frac{1 - 95\%}{ 1-p_k}
\end{align}
With $m = 5000 \ln(5/\delta)$, we can now bound:
    \begin{align}
        \text{Pr}[\text{terminate at } k] &\leq \exp( - m D( 95\% || p_k )  ) \\
        &= \left[   \left(\frac{95\%}{p_k}\right)^{-95\%} \left( \frac{5\%}{1-p_k} \right)^{-5\%} \right]^m\\
        &\leq  \left[  (95\%)^{-95\%} (5\%)^{-5\%}  \right]^m \cdot p_k^{m\cdot 95\%} \\
        &\leq  1.22^m \cdot p_k^{m\cdot 95\%} 
    \end{align}
    Next, we plug in $p_k = \sin^2(r_k \theta) \leq (r_k \theta)^2 \leq (\theta \cdot 1.05^k )^2$, and use the union bound to bound the probability that we terminate at or before $k_0$:
    \begin{align}
        \text{Pr}[\text{terminate} \leq k_0] &= \sum_{k=0}^{k_0} \text{Pr}[\text{terminate at } k] \\ 
        &\leq  1.22^m \cdot  \sum_{k=0}^{k_0} (\theta \cdot 1.05^k )^{1.9 m} \\
       \sum_{k=0}^{k_0} (\theta \cdot 1.05^k )^{1.9 m}  &= \theta^{1.9 m}  \cdot \sum_{k=0}^{k_0} (1.05^{1.9 m})^k\\
        &= \theta^{1.9 m}  \cdot  \frac{(1.05^{1.9 m})^{k_0 + 1} - 1}{1.05^{1.9 m} - 1} \\
        &\leq \left( \theta \cdot 1.05^{k_0} \right)^{1.9 m}  \cdot  \frac{1.05^{1.9 m}}{1.05^{1.9 m} - 1} \\
        &\leq  0.9^{1.9 m}  \cdot 1.01 \\
        \text{Pr}[\text{terminate} \leq k_0] &\leq  \left[ 1.22 \cdot 0.9^{1.9} \right]^m \cdot 1.01\\
        &\leq  \left[ 0.9984 \right]^{ 5000 \ln(5/\delta) } \cdot 1.01\\
        &\leq  \left[ e^{-1} \right]^{\ln(5/\delta) } \cdot 1.01 \\
        &=  e^{-\ln(5/\delta)} = (\delta/5) \cdot 1.01   \leq \delta/4
    \end{align}
    The probability we terminate too early is $\leq \delta/4$, and the probability we terminate too late is $\leq \delta/5$, so the probability that the first step has the ideal behavior is $\geq 1-\delta/2$.

 \noindent    \textit{Query complexity of step 1.} When the first step behaves ideally, then it terminates at $k_\text{end} = k_0 + 10$ at the latest. Since another way of writing $k_0$ is $k_0 =\lfloor  \log_{1.05}(0.9/\theta) \rfloor $, we have
    \begin{align}
        \sum_{k=0}^{k_\text{end}} r_k \leq \sum_{k=0}^{k_0 + 10} 1.05^k \leq \frac{ 1.05^{k_0 + 11} -1  }{1.05-1} \in O( 1.05^{k_0}) \leq O( \theta^{-1} ) 
    \end{align}
        Since at each iteration we run $O(\log(\delta^{-1}))$ circuits, the total complexity in the ideal case is $O( \theta^{-1} \log(\delta^{-1}) ) = O(\sqrt{N/K} \log(\delta^{-1}))$.

 \noindent   \emph{Correctness of step 2.} Let $\gamma := \theta_\text{max} / \theta_\text{min} - 1$. \ By definition, the $\theta_\text{min}, \theta_\text{max}$ initialized at the beginning of step 2 satisfy $\gamma  = 1.65  -1  = 0.65 \leq 0.7 $. \ We also see that $\theta_\text{max} \leq 1.65\theta \leq 0.002$, so we satisfy the conditions of Lemma~\ref{lemma:rotation}. The coin described in the lemma is implemented by measuring the state $G^{(r-1)/2}|\psi\rangle$.

    Each iteration of step 2 will modify $\theta_\text{min}$ or $\theta_\text{max}$ in order to reduce $\gamma$ by exactly a factor of $0.9$. \ Each iteration preserves $\theta_\text{min} \leq \theta \leq \theta_\text{max}$ with high probability.  When the algorithm terminates we have $\theta_\text{max}/\theta_\text{min} \leq 1 + \eps/5$ which implies that any value $\hat\theta$ between $\theta_\text{min}$ and $\theta_\text{max}$ satisfies $(1- \eps/5)\theta \leq \hat\theta \leq (1+ \eps/5)\theta$ as desired. \ A simple calculation\footnote{Show $ 1 - \eps \leq \sin^2(\theta (1+\eps/5)) / \sin^2(\theta) \leq 1+\eps$ by Taylor expanding  $\sin^2(\theta (1+\eps/5)) / \sin^2(\theta) $ around $\eps = 0$ and truncating the series to obtain bounds.} shows that these multiplicative error bounds on the estimate for $\theta$ guarantee the desired  $(1-\eps)K \leq \hat K \leq (1+\eps)K$.

    We start with $t= 0$, so at the beginning of step $t$ we have $ \gamma_t = 0.65 \cdot 0.9^{t}$. Let $T$, the iteration after which we terminate in step 2, be given by:
    \begin{align}
        T := \text{the largest integer satisfying } \gamma_T = 0.65 \cdot (0.9)^{T} \geq \frac{\eps}{5}
    \end{align}
   This way, after the $T$'th step we have $\gamma_{T+1} \leq \eps/5$, so we stop. (Note that there are $T+1$ total iterations.) For convenience we define $b := 1/0.9 \approx 1.11$. Note that $b >1$, which makes $\log_b(x)$ behave intuitively. Then: 
    \begin{align}
        T \leq \log_{b} \left( \frac{3.25}{\eps} \right) 
    \end{align}

  We used Lemma~\ref{lemma:rotation} to guarantee that the failure probability at the $t$'th  iteration is at most $\delta_t := (\delta \eps / 65) \cdot (0.9)^{-t} = (\delta\eps / 65) \cdot b^t$. \ By the union bound, the overall failure probability is then at most:
    \begin{align}
        \sum_{t=0}^T \delta_t = \frac{\delta \eps}{65} \sum_{t=0}^T b^t \leq \frac{\delta \eps}{65} \frac{b^{T+1} - 1}{b-1} \leq \frac{\delta \eps}{65} \frac{b}{b-1} \frac{3.25}{\eps} = \frac{\delta}{2}
    \end{align}

  \noindent  \emph{Query complexity of step 2.} From Lemma~\ref{lemma:rotation}, we know that $r_t \in O(\theta^{-1}\gamma_t^{-1})$. Recall that $\gamma_t = 0.65 \cdot b^{-t}$. Each step makes $O(\ln(\delta_t^{-1}))$ queries. So the total query complexity is given by:
    \begin{align}
        O\left( \sum_{t=0}^T \frac{1}{\theta} \frac{1}{\gamma_t} \ln\left( \frac{1}{\delta_t} \right)  \right) &\leq     O\left( \frac{1}{\theta} \sum_{t=0}^T b^t \left[\ln\left( \frac{1}{\delta} \right) + \ln\left( \frac{1}{\eps} \right) - t \ln\left( b\right) \right] \right)\\
                    &\leq     O\left( \frac{1}{\theta} b^T \ln\left( \frac{1}{\delta} \right) + \frac{1}{\theta} \sum_{t=0}^T b^t \left[ \ln\left( \frac{1}{\eps} \right) - t \ln\left( b\right) \right] \right)\\
                    &\leq     O\left( \sqrt{\frac{N}{K}} \frac{1}{\eps} \ln\left( \frac{1}{\delta} \right) \right) + O\left(\frac{1}{\theta} \sum_{t=0}^T b^t \left[ \ln\left( \frac{1}{\eps} \right) - t \ln\left( b\right) \right] \right)
    \end{align}
The first term is the desired complexity, so all that remains to show is that the second term is dominated by the first term. To do so, we compute some bounds on $T$:
\begin{align}
    T +1&\geq\log_b \left( \frac{3.25}{\eps} \right) \\
    \ln(b) (T+1) &\geq \ln(b) \frac{\ln\left( \frac{3.25}{\eps} \right)}{\ln(b)} = \ln\left( \frac{3.25}{\eps} \right) \geq \ln\left( \frac{1}{\eps} \right)
\end{align}
Now we bound the second term, dropping the $1/\theta$. To aid intuition, recall that $b > 1$.
\begin{align}
    \sum_{t=0}^T b^t \left[ \ln\left( \frac{1}{\eps} \right) - t \ln\left( b\right) \right] &= \ln\left(\frac{1}{\eps}\right) \frac{b^{T+1}-1}{b-1}  - \ln\left(b\right) \frac{b}{(b-1)^2} \left[ T b^{T+1} - (T+1)b^T + 1\right]\\
    &\leq \ln\left(\frac{1}{\eps}\right) \frac{b^{T+1}}{b-1}  + \ln\left(b\right) \frac{b}{(b-1)^2} \left[- T b^{T+1} + (T+1)b^T \right]\\
    &= \ln\left(\frac{1}{\eps}\right) \frac{b^{T+1}}{b-1}  + \ln\left(b\right) \frac{b}{(b-1)^2} \left[b^{T+1} - (T+1) b^{T+1} + (T+1)b^T  \right]\\
    &= \ln\left(\frac{1}{\eps}\right) \frac{b^{T+1}}{b-1}   + \ln\left(b\right) \frac{b}{(b-1)^2} \left[  b^{T+1} - (T+1) b^{T}(b-1) \right]\\
    &= \ln\left(\frac{1}{\eps}\right) \frac{b^{T+1}}{b-1}   - \ln\left(b\right)(T+1) \frac{b^{T+1}}{b-1} + \ln\left(b\right) \frac{b^{T+2}}{(b-1)^2} \\
    &\leq \ln\left(\frac{1}{\eps}\right) \frac{b^{T+1}}{b-1}   - \ln\left(\frac{1}{\eps}\right) \frac{b^{T+1}}{b-1} + \ln\left(b\right) \frac{b^{T+2}}{(b-1)^2} \\
    &= \ln\left(b\right) \frac{b^{T+2}}{(b-1)^2} \in   O\left( \frac{1}{\eps} \right)
\end{align}
\end{proof}

We note that this algorithm can also be used to determine if there are no marked items. If there is at least one marked item, then $\theta \geq \arcsin\sqrt{1/N}$. That means, that there must be some point in step 1(b) where we are likely to see enough marked items. If we fail to see enough marked items then, then we must have $K = 0$ with high probability.

Next we prove Lemma~\ref{lemma:rotation}, which constructs a number of rotations $r$ such that when $\theta \approx \theta_\text{max}$ it is very likely to see a marked item, and when $\theta \approx \theta_\text{min}$ it is very unlikely to see a marked item.


\begin{lemma}\label{lemma:rotation} \textbf{Rotation lemma.} Say we are given $\theta_\text{min},\theta_\text{max}$ such that $0<\theta_\text{min}\leq\theta\leq\theta_\text{max}\leq 0.002$ and  $\theta_\text{max} = ( 1 + \gamma)\cdot \theta_\text{min}$ for some $\gamma \leq 0.7$. \ Then we can calculate an odd integer $r\in O( \theta^{-1}\gamma^{-1})$ such that the following is true: Consider tossing a coin that is heads with probability $\sin^2(r\theta)$ at least $250 \cdot \ln(\delta^{-1})$ times, and subsequently

\begin{enumerate}
    \item if more than $12\%$ heads are observed set $\theta_\text{min} $ to $ \theta_\text{max} / (1+0.9\gamma)$,
    \item and otherwise set $\theta_\text{max} $ to $(1+0.9\gamma)\theta_\text{min}$.
\end{enumerate}
    This process fails to maintain $\theta_\text{min} \leq \theta \leq \theta_\text{max}$ with probability less than $\delta$. 
\end{lemma}


\begin{proof} We compute $r$ as follows (when rounding, ties between integers can be broken arbitrarily):
    \begin{align}
        \Delta\theta &:= \theta_\text{max} - \theta_\text{min}\\
        k &:= \text{ the closest integer to } \frac{\theta_\text{min}}{2\Delta\theta}\\
        r &:= \text{ the closest odd integer to } \frac{\pi k}{\theta_\text{min}}
    \end{align}
   First we show that $r\theta_\text{min} \approx \pi k$ and $r\theta_\text{max} \approx \pi k + \frac{\pi}{2}$. We notice some basic facts about $\Delta\theta$ following from $\theta_\text{max} = ( 1 + \gamma)\cdot \theta_\text{min}$:
    \begin{align}
        \frac{\Delta\theta}{\theta_\text{min}} &= \gamma \\
        \frac{\Delta\theta}{\theta_\text{max}} &= 1 - \frac{1}{1+\gamma}
    \end{align}
Now we bound the rounding errors on $k$ and $r$:
    \begin{align}
        \left|k   - \frac{\theta_\text{min}}{2\Delta\theta} \right| &\leq \frac{1}{2}\\
         \left|r   - \frac{\pi k}{\theta_\text{min}} \right| &\leq 1
    \end{align}
    We arrive with bounds on $r\theta_\text{min}$ and $r\theta_\text{max}$:
    \begin{align}
        \left| r\theta_\text{min}   - \pi k \right| &\leq \theta_\text{min} \leq 0.002\\
        \left| \pi k \frac{\Delta\theta}{\theta_\text{min}}   - \frac{\pi}{2} \right| &\leq \frac{\pi}{2} \cdot \frac{\Delta\theta}{\theta_\text{min}} = \frac{\gamma\pi}{2} \\
        \left| r\theta_\text{max} - \left(\pi k + \frac{\pi}{2}\right) \right| &\leq \left| \pi k \frac{\theta_\text{max}}{\theta_\text{min}} - \left(\pi k + \frac{\pi}{2}\right) \right| + \theta_\text{max}\\
        &= \left| \left(\pi k + \pi k \frac{\Delta\theta}{\theta_\text{min}}\right) - \left(\pi k + \frac{\pi}{2}\right) \right| + \theta_\text{max}\\
        &\leq  \frac{\gamma \pi}{2}  + \theta_\text{max}\\
        &\leq 0.7 \cdot \frac{\pi}{2} + 0.002 \leq 1.102
    \end{align}
    Given these bounds, we can examine the two cases when we fail to preserve $\theta_\text{min} \leq \theta \leq \theta_\text{max}$ and demonstrate that they are unlikely: we could see many heads even though $\theta$ is small, thus fail to preserve $\theta_\text{min} \leq \theta$, or we could see few heads even though $\theta$ is large, and thus fail to preserve $\theta \leq \theta_\text{max}$.

    First, suppose that $\theta_\text{min} \leq \theta \leq \theta_\text{max} / (1 + 0.9\gamma)$, so $\theta$ is near the bottom of the interval. We are to show that it is unlikely that we see too many heads.
    \begin{align}
        r\theta &\geq r\theta_\text{min} \geq \pi k - 0.002\\
        r\theta &\leq \frac{r\theta_\text{max}}{1+ 0.9\gamma}\\
                &\leq r\theta_\text{max} - \left( r\theta_\text{max}  - \frac{r\theta_\text{max}}{1+0.9\gamma}  \right)\\
                &\leq r\theta_\text{max} - r \Delta\theta \frac{\theta_\text{max}}{\Delta\theta} \left( 1 - \frac{1}{1+0.9\gamma}  \right)\\
                &\leq r\theta_\text{max} - r \Delta\theta  \frac{ 1 - \frac{1}{1+0.9\gamma} }{ 1 - \frac{1}{1 + \gamma} }
    \end{align}
    We find $ \frac{ 1 - \frac{1}{1+0.9\gamma} }{ 1 - \frac{1}{1 + \gamma} } \geq 0.9$: viewing the left hand side as a function of $\gamma$, this function increases with $\gamma$, so the smallest value occurs as $\gamma \to 0$. Proceeding with the upper bound on $r\theta$:
    \begin{align}
         r\theta&\leq r\theta_\text{max} - 0.9 \cdot r \Delta\theta  \\
                &\leq 0.1 \cdot r\theta_\text{max} + 0.9 \cdot r \theta_\text{min}\\
                &\leq \pi k +  0.1 \cdot \left(\frac{\pi}{2}  + 1.102\right)  + 0.9 \cdot 0.002 \\
                &\leq \pi k +  0.27.
    \end{align}
    Since $-0.002 \leq r\theta - \pi k \leq 0.27$, we have $\sin^2(r\theta) \leq 7.5\%$. So, the probability we fail is bounded by:
    \begin{align}
        \text{Pr}[\text{see more than } 12\% \text{ heads}] \leq \exp(- 2 \cdot (12\% - 7.5\%)^2 \cdot 250 \ln(\delta^{-1})   ) \leq \delta
    \end{align}

    Now, we consider the case $\theta_\text{min} (1+0.9\gamma) \leq \theta \leq \theta_\text{max} $ where $\theta$ is near the top of the interval. We must show that it is unlikely that we see too few heads.
    \begin{align}
        r\theta &\leq r\theta_\text{max} \leq \pi k + \frac{\pi}{2} + 1.102 \leq \pi k + 2.68\\
        r\theta &\geq r\theta_\text{min}(1+0.9\gamma) \\
                &= r\theta_{min} + 0.9 \cdot r\Delta\theta\\
                &= 0.9 \cdot r\theta_\text{max} + 0.1 \cdot r\theta_\text{min}\\
                &\geq 0.9 \cdot \left( \pi k + \frac{\pi}{2} - 1.102  \right) + 0.1 \cdot \left(\pi k - 0.002\right)\\
                &\geq \pi k + 0.9 \cdot \frac{\pi}{2} - 0.9\cdot1.102 - 0.1\cdot 0.002\\
                &\geq \pi k + 0.42.
    \end{align}
    Since $0.42 \leq r\theta - \pi k \leq 2.68$ we have $\sin^2(r\theta) \geq 16.5\%$.  So, the probability we see too few heads is bounded by:
    \begin{align}
        \text{Pr}[\text{see fewer than } 12\% \text{ heads}] \leq \exp(- 2 \cdot (16.5\% - 12\%)^2 \cdot 250 \ln(\delta^{-1})   ) \leq \delta
    \end{align}

    We have shown correctness. Finally, we see that:
    \begin{align}
        r \in O\left( \frac{k}{\theta_\text{min}} \right) \leq O\left( \frac{1}{\Delta\theta}  \right) \leq O\left( \frac{1}{\gamma\theta} \right).
    \end{align}
\end{proof}

\section{Amplitude Estimation\label{AMPEST}}

We now show how to generalize our algorithm for approximate counting to amplitude estimation: given two quantum states $|\psi\rangle$ and $|\phi\rangle$, estimate the magnitude of their inner product $a = \left|\langle\psi|\phi\rangle\right|$. \ We are given access to these states via a unitary $U$ that prepares $|\psi\rangle$ from a starting state $|0^n\rangle$, and also marks the component of $|\psi\rangle$ orthogonal to $|\phi\rangle$ by flipping a qubit. \ Recall that our analysis of approximate counting was in terms of the `Grover angle' $\theta := \arcsin\sqrt{K/N}$. \ By redefining $\theta := \arcsin a$, the entire argument can be reused.

\begin{theorem} \label{thm:ampest} Given  $\eps,\delta > 0$ as well as access to an $(n+1)$-qubit unitary $U$ satisfying
    $$U|0^n\rangle|0\rangle = a |\phi\rangle |0\rangle + \sqrt{1-a^2}|\tilde\phi\rangle|1\rangle,$$
    where $|\phi\rangle$ and $|\tilde\phi\rangle$ are arbitrary $n$-qubit states and $0<a<1$,\footnote{Note that we can always make $a$ real by absorbing phases into $|\phi\rangle,|\tilde\phi\rangle$.} there exists an algorithm that outputs an estimate $\hat a$ that satisfies
    $$  a(1-\eps) < \hat a < a(1+\eps)  $$
    and uses $O\left( a^{-1} \eps^{-1} \log(\delta^{-1}) \right)$ applications of $U$ and $U^\dagger$ with probability at least $1-\delta$.
\end{theorem}

\begin{proof} The algorithm is as follows.\\

    \noindent\fbox{\parbox{\textwidth}{\hspace{3mm}\parbox{0.95\textwidth}{\hspace{1mm}\vspace{2mm}\\ 
    \textbf{Algorithm: Amplitude Estimation}\\
    \textbf{Inputs:} $\eps,\delta > 0$ and a unitary $U$ satisfying
        $U|0^n\rangle|0\rangle= a |\phi\rangle |0\rangle + \sqrt{1-a^2}|\tilde\phi\rangle|1\rangle$.\\
    \textbf{Output:} An estimate of $a$.\\

    Let $R$ satisfy $R|0\rangle = \frac{1}{1001} |0\rangle + \sqrt{1-\left(\frac{1}{1001}\right)^2}|1\rangle$. Then:
    \begin{equation}U|0^n\rangle|0\rangle \otimes R|0\rangle = \frac{ a}{1001} |\phi\rangle|00\rangle + \text{ terms orthogonal to }|\phi\rangle|00\rangle \end{equation}
    Define $\theta :=\arcsin \frac{a}{1001}$ and we have $0 \leq \theta \leq 0.001$. Let the Grover diffusion operator $G$ be:
    \begin{equation}G := -(U\otimes R)(I -2 |0^{n+2}\rangle\langle 0^{n+2}|)(U\otimes R)^\dagger(I_{n+2} - 2(I_n\otimes |00\rangle\langle 00|))\end{equation}

    \begin{enumerate}
        \item Follow steps 1 and 2 in the algorithm for approximate counting. An item is `marked' if the final two qubits are measured as $|00\rangle$.
        \item Return $\hat a := 1001\cdot\sin\left(\theta_\text{max}\right)$ as an estimate for $a$.
    \end{enumerate}
    }}}\\

    If we write
    \begin{equation}(U\otimes R)|0^{n+2}\rangle = \sin\theta |\phi00\rangle + \cos\theta |\phi00^\perp\rangle \end{equation}
    where $|\phi00^\perp\rangle$ is the part of the state orthogonal to $|\phi00\rangle$, then the Grover operator $G$ rotates by an angle $2\theta$ in the two-dimensional subspace spanned by $\{|\psi00\rangle, |\psi00^\perp\rangle\}$. \ Therefore:
    \begin{equation}G^{(r-1)/2}(U\otimes R)|0^{n+2}\rangle = \sin(r\theta) |\phi00\rangle + \cos(r\theta) |\phi00^\perp\rangle \end{equation}
     making the probability of observing $|00\rangle$ on the last two qubits equal to $\sin^2(r\theta)$. This is the quantity bounded by Lemma~\ref{lemma:rotation}.

    The remainder of the proof is identical to the one for approximate counting, which guarantees that $\theta_\text{max}$ is an estimate of $\theta$ up to a $1+\eps/5$ multiplicative factor. \ A simple calculation shows that $1001\cdot\sin\left(\theta_\text{max}\right)$ is then an estimate of $a$ up to a $1+\eps$ multiplicative factor.
\end{proof}

\section{Open Problems\label{OPEN}}

What if we limit the ``quantum depth'' of our algorithm? \ That is, suppose we assume that after every $T$ queries, the algorithm's quantum state is measured and destroyed. \ Can we derive tight bounds on the quantum query complexity of approximate counting in that scenario, as a function of $T$? \ This question is of potential practical interest, since near-term quantum computers {\em will} be severely limited in quantum depth. \ It's also of theoretical interest, since new ideas seem needed to adapt the polynomial method \cite{bbcmw} or the adversary method \cite{ambainis} to the depth-limited setting (for some relevant work see \cite{jmw}).

Can we do approximate counting with the optimal $O\left( \eps^{-1} \sqrt{N/K}\right)
$ quantum query complexity (and ideally, similar runtime), but with Grover invocations that are parallel rather than sequential, as in the algorithm of Suzuki et al. \cite{suzuki}?

\section*{Acknowledgments}

We thank Paul Burchard for suggesting the problem to us, as well as John Kallaugher, C.\ R.\ Wie, Ronald de Wolf, Naoki Yamamoto, Xu Guoliang, Seunghoan Song, Hanqing Wu, and Sabee Grewal for helpful discussions. We especially thank Justin Yirka for carefully reading the manuscript and catching several errors.

\bibliographystyle{alpha}
\bibliography{oct7_noqft}

\end{document}